\documentclass[11pt]{article}

\usepackage{titling}

\usepackage{amsmath,amssymb,amsthm}
\usepackage{mathabx}
\usepackage{caption}
\usepackage{color}
\usepackage{tikz}
\usetikzlibrary{arrows}
\usepackage{verbatim}
\usepackage{palatino}
\usepackage{multirow}
\usepackage{mathpazo}
\usepackage{stmaryrd}
\usepackage[margin=1in]{geometry}
\usepackage{mathtools}
\mathtoolsset{centercolon}
\usepackage{thm-restate}
\usepackage{hyperref}
\usepackage{tikz}	
\usetikzlibrary{backgrounds,fit,decorations.pathreplacing}  
\usetikzlibrary{arrows}
\usepackage{scalefnt}

\definecolor{darkred}  {rgb}{0.5,0,0}
\definecolor{darkblue} {rgb}{0,0,0.5}
\definecolor{darkgreen}{rgb}{0,0.5,0}
\hypersetup{
  pdftitle = {Cryptographic Reversibility Notes},
  pdfauthor = {Eric Chitambar},
  colorlinks = true,
  urlcolor  = blue,         
  linkcolor = darkblue,     
  citecolor = darkgreen,    
  filecolor = darkred       
}

\newtheorem{theorem}{Theorem}
\newtheorem{proposition}{Proposition}
\newtheorem{lemma}{Lemma}

\newtheorem{corollary}{Corollary}
\theoremstyle{definition}
\newtheorem{definition}{Definition}
\newtheorem*{remark}{Remark}

\newcommand{\bra}[1]{\langle #1|}
\newcommand{\ket}[1]{|#1\rangle}

\newcommand{\ip}[2]{\langle #1|#2\rangle}
\newcommand{\op}[2]{|#1\rangle \langle #2|}

\DeclareMathOperator{\tr}{Tr}

\newcommand{\mc}[1]{\mathcal{#1}}
\newcommand{\mbf}[1]{\mathbf{#1}}

\title{Quantum Versus Classical Advantages in Secret Key Distillation (and Their Links to Quantum Entanglement) \vspace{-.2cm}}

\author{Eric Chitambar$^1$, \quad Ben Fortescue$^1$, \quad  Min-Hsiu Hsieh$^2$
\\[4mm]
\textit{$^1$ Department of Physics and Astronomy, Southern Illinois University,}\\ 
\textit{Carbondale, Illinois 62901, USA}\\
\textit{$^2$ Centre for Quantum Computation \& Intelligent Systems (QCIS),}\\
\textit{Faculty of Engineering and Information Technology (FEIT),}\\
\textit{University of Technology Sydney (UTS), NSW 2007, Australia}}

\date{\today}

\begin{document}
\maketitle

\vspace{-1cm}

\begin{abstract}

We consider the extraction of shared secret key from correlations that are generated by either a classical or quantum source. In the classical setting, two honest parties (Alice and Bob) use public discussion and local randomness to distill secret key from some distribution $p_{XYZ}$ that is shared with an unwanted eavesdropper (Eve).  In the quantum settings, the correlations $p_{XYZ}$ are delivered to the parties as either an \textit{incoherent} mixture of orthogonal quantum states or as \textit{coherent} superposition of such states; in both cases, Alice and Bob use public discussion and local quantum operations to distill secret key.  While the power of quantum mechanics increases Alice and Bob's ability to generate shared randomness, it also equips Eve with a greater arsenal of eavesdropping attacks.  Therefore, it is not obvious who gains the greatest advantage for distilling secret key when replacing a classical source with a quantum one.  

\smallskip

In this paper we first demonstrate that the classical key rate is equivalent to the quantum key rate when the correlations are generated incoherently in the quantum setting.  For coherent sources, we next show that the rates are incomparable, and in fact, their difference can be arbitrarily large in either direction. However, we identify a large class of non-trivial distributions $p_{XYZ}$ that possess the following properties: (i) Eve's advantage is always greater in the quantum source than in its classical counterpart, and (ii) for the quantum entanglement shared between Alice and Bob in the coherent source, the so-called entanglement cost/squashed entanglement/relative entropy of entanglement can all be computed.  With property (ii), we thus present a rare instance in which the various entropic entanglement measures of a quantum state can be explicitly calculated. 

\end{abstract}

\vspace{-.375cm}

\section{Introduction}

When possessing a shared secret key, two parties can communicate over a public communication channel in a provably secure manner.  Specifically, by using the key for a one-time pad encryption, the public message can faithfully encode the secret message, thereby protecting it from any eavesdropping third party.  While it is impossible to establish secret key using public communication alone, it turns out that public communication can often be used to transform partially secret key into a stronger, more usable form, a process known as \textit{secret key agreement} or \textit{secret key distillation} \cite{Maurer-1993a, Ahlswede-1993a}.  More precisely, suppose that the two parties (called Alice and Bob) share correlated random variables $X$ and $Y$, but an eavesdropper (called Eve) has partial knowledge contained in her variable $Z$.  Using public communication and local processing of their variables, Alice and Bob may be able to generate a highly correlated pair of variables $(\hat{X},\hat{Y})$ that is strongly uncorrelated with both $Z$ and the conducted public communication, i.e. secret key. 

Typically one considers the scenario where Alice, Bob, and Eve have access to some source that generates identical and independent copies of $XYZ$.  The figure of merit then becomes the \textit{secret key rate}, which is the largest number of secret bits per copy that Alice and Bob can obtain using local processing and public communication.  Clearly if $Z$ is highly correlated with either $X$ or $Y$, then it will be impossible for key to be generated, even in the many-copy setting.  But it is a significant open problem to understand precisely how Eve's side information $Z$ affects the key rate.

There are a variety of physical situations in which one might encounter a many-copy source of variables $XYZ$.  Most notably is the task of quantum key distribution (QKD) in which the variables $XYZ$ are generated through the inherently stochastic nature of quantum measurement \cite{Bennett-1984a}.  Alice, Bob, and Eve share a tripartite quantum state of the form $\ket{\Psi}^{ABE}=\sum_{x,y,z}e^{i\varphi_{xyz}}\sqrt{p(x,y,z)}\ket{xyz}^{ABE}$, where $\varphi_{xyz}$ are arbitrary phases and $p(x,y,z)$ describes a joint distribution for variables $XYZ$.  When the three parties measure their quantum system in the computational basis (i.e. in the $\{\ket{x}^A\}$, $\{\ket{y}^B\}$, and $\{\ket{z}^E\}$ basis respectively), their measurement outcomes are distributed according to $p(x,y,z)$.  If this is done on multiple copies of $\ket{\Psi}^{ABE}$, the parties thus generate a many-copy source of $XYZ$ from which Alice and Bob can distill secret key using public discussion and local processing.

Note that the described scenario only describes one particular way that Alice and Bob could use multiple copies of $\ket{\Psi}^{ABE}$ to obtain key.  With quantum mechanics, more physical operations are allowed than just measuring in the computational basis.  They could, for instance, put their local subsystems through some quantum channel (i.e. a trace-preserving, completely positive map), or they could each entangle their local subsystems and subject them to some joint unitary evolution before measuring.  Naively then, it appears that with such greater operational powers, Alice and Bob can always distill at least as much key from a quantum source of $\ket{\Psi}^{ABE}$ than from a classical source of the underlying random variables $XYZ$.  However, in the quantum scenario, Eve also gains operational strength in her eavesdropping abilities.  This begs the natural question: for the purpose of secret key distillation, who gains the greatest advantage when embedding a given distribution $p(x,y,z)$ into a multi-party quantum system, the honest parties or the adversary?  

Answering this question is the central aim of this paper.  Through the construction of specific examples, we show that the advantage can lie either with Alice/Bob or with the Eve.  Hence the adage ``quantum is more powerful than classical'' is really a matter of perspective when it comes to the task of secret key distillation.  Furthermore, we prove that \textit{quantum coherence} plays the essential role in affecting whether the quantum key rate differs from its classical counterpart.  More precisely, in the state $\ket{\Psi}^{ABE}$ given above, the distribution $p(x,y,z)$ is encoded as a coherent superposition of the basis states $\ket{xyz}^{ABE}$.  An alternative form of quantum embedding is an \textit{incoherent} mixture of states $\rho^{ABE}=\sum_{x,y,z}p(x,y,z)\op{xyz}{xyz}$.  We prove that even when Alice and Bob are allowed to perform arbitrary quantum processing on their respective parts of $\rho^{ABE}$, their optimal rate of key extraction is not improved over the corresponding classical key rate.  This result is significant as it identifies quantum coherence as precisely the ingredient that distinguishes classical from quantum secret key distillation, something that has not previously been understood.  

\begin{table}[bp] 
\begin{tabular*}{\hsize}{@{\extracolsep{\fill}}cc}
\hline
Rates & Results Presented Here\\ \hline
$K_D(p_{XYZ})$ vs. $K_D(\rho_{ccc})$ & $K_D(p_{XYZ})=K_D(\rho_{ccc})$\\
\hline
\multirow{2}{*}{} $K_D(p_{XYZ})$ vs. $K_D(\Psi_{qqq})$ & no general relationship \\&with arbitrarily large gaps existing in both directions\\ \hline
\multirow{3}{*}{}$K_D(\rho_{ccq})$ vs.&\\$K_D(\rho_{cqq})$ vs. &$K_D(\rho_{ccq})\leq K_D(\rho_{cqq})\leq K_D(\rho_{qqq})$\\$K_D(\rho_{qqq})$  & with arbitrarily large gaps existing between the rates\\\hline
\end{tabular*}\vspace{3mm}
\caption{Quantum Versus Classical Advantages in Secret Key Distillation.}
\label{Table:Summary}
\end{table}

Our second main result involves computing the entanglement of a quantum state based on the properties of its embedded classical distribution.  Evaluating some of the most important entanglement measures for a general quantum state is a notoriously difficult problem due to the variational character of these measures.  However, as we will demonstrate in this paper, when embedding quantum states with certain types of probability distributions, the entanglement can be bounded by the secret key rate of the underlying distribution; and in some cases the two are equivalent.  This offers a remarkable demonstration of how cryptographic results in classical information can be used to uncover novel physical properties of quantum systems.

The structure of this paper is as follows. In Sec.~\ref{Sec_Results}, we briefly summarize the main results of the paper.  We then move to Sec.~\ref{Sec_2} where a relatively self-contained overview of the necessary concepts are present.  In particular, we describe a unified framework for local information processing and public communication in both classical and quantum key distillation protocols.  Secret key is presented as a classical analog to quantum entanglement, and the tasks of secret key distillation/formation are described as the counterparts to entanglement distillation/formation.  While still in Sec.~\ref{Sec_2}, we define classes of probability distributions that possess special properties such as the ability to dilute and compress secret correlations at equal rates.  Main results and their proofs are collected in Sec.~\ref{sec_Main}. 

\section{Summary of Results}
\label{Sec_Results}

\begin{table*}[bp] 
\begin{tabular*}{\hsize}{@{\extracolsep{\fill}}lcc|}
\hline
{Type of Distribution} & Results Presented Here \\ \hline
  \multirow{2}{*}{Secrecy Reversible} & \multirow{1}{*}{$K_D(p_{XYZ})\geq E_{sq}(\rho^{AB})$}\\ 
  &  \\ \hline
  Secrecy Reversible & \multirow{ 1}{*}{$K_D(p_{XYZ})\geq E_F(\rho^{AB})$}\\
 + UBI-PD &  \\ \hline
Secrecy Reversible &  \multirow{1}{*}{$K_C(p_{XYZ})=E_{sq}(\rho^{AB})=K_D(\Psi^{ABE})=K_D(p^{XYZ})$} \\
 + Semi-unambiguous & \\ \hline
Secrecy Reversible &\multirow{1}{*}{$K_C(p_{XYZ})=E_{F}(\rho^{AB})=E_C(\rho^{AB})=E_{sq}(\rho^{AB})=E_{r}(\rho^{AB})$}  \\
+UBI-PD & \multirow{1}{*}{$=E_D(\rho^{AB})=K_D(\Psi^{ABE})=K_D(p_{XYZ})=H(J_{XY|Z}|Z)$}\\
 + Semi-Unambiguous &    \\ \hline
\end{tabular*}
\vspace{3mm}
\caption{Linking secret keys to quantum entanglement.}
\label{Table:Summary2}
\end{table*}

We now summarize our results in two tables.  The first describes classical versus quantum key rates $K_D$.  The underlying classical distribution is $p_{XYZ}$ and $\ket{\Psi_{qqq}}$ is its coherent embedding in a tripartite state.  When one or more of the parties dephases in the computational basis, it generates a mixed state, and in this state we make the notational change $q\to c$ for the corresponding parties who dephase.  Thus $\rho_{ccc}$ is a full incoherent mixture according to the classical distribution $p_{XYZ}$.  For completion we note a result given by Christandl \textit{et. al} who showed a distribution $p_{XYZ}$ for which Eve being quantum gives her a definite increase in power $K_D(\rho_{ccc})>K_D(\rho_{ccq})$  \cite{Christandl-2007a}.

Our second series of results involves analyzing the quantum versus classical rates for special classes of distributions.  Details of these classes are presented in Sect \ref{Sec_2}, and here we just state the results.  Most notably are the distributions possessing so-called secrecy reversibility, which can be shown to have the following property related to quantum versus classical advantages.
\begin{theorem}
\label{thm:ReversibilityAdvantages}
If a distribution $p_{XYZ}$ possess secrecy reversibility, then Alice and Bob can gain no advantage over Eve when embedding their correlations into a quantum source (i.e. $K_D(p_{XYZ})\geq K_D(\Psi_{qqq})$).
\end{theorem}
It turns out that for a special subclass of reversible distributions, which also belong to the family of so-called semi-unambiguous distributions, we are able to compute the quantum entanglement in the embedded state.  Specifically, the squashed entanglement ($E_{sq}$), the relative entropy of entanglement ($E_{r}$), and the entanglement of formation ($E_F$) can all be computed.  This is remarkable since all previous calculations of these measures for a given mixed state involve exploring certain symmetries of the state (such as the so-called quantum flower states \cite{1499048}).  Our results are obtained purely by relating entanglement to the classical problem of secrecy reversibility.  Our results are summarized in Table \ref{Table:Summary2}, where the quantity $J_{XY|Z}$ is the conditional common information .

\section{Preliminary}
 \label{Sec_2}
 
\subsection{A Unified Framework for Local Information Processing and Public Communication}

In this section we review the definitions of classical and quantum local operations and public communication (LOPC).

\subsubsection{Classical Operations.}

In the classical LOPC setting, each party is allowed to perform the following operations:
\begin{enumerate}
\item[(i)] Generate local random variables that are uncorrelated from the variables held by any other party.
\item[(ii)] Generate copies of any locally held variables.
\item[(iii)] Change the values of any locally held variables according to some function.
\item[(iv)] Broadcast the result of any computed function over an authenticated public channel.
\end{enumerate}
Note that operations (i) - (iii) encompass any sort of noisy processing that a party may wish to perform.  A general \textbf{classical LOPC protocol} $\mc{P}_c$ then consists of two phases: Phase I - a coordinated and multi-round exchange of public messages in which each message is a function of some party's local variables, and Phase II - each party processes his/her variables through a local channel chosen according to the messages of Phase I, thereby generating the output variables of the protocol.  It is not difficult to see that the variables generated by any sequence of operations (i)--(iv) can always also be generated by a protocol following this two-phase format \cite{Ahlswede-1993a, Maurer-1993a, Csiszar-2011a}.

Consider now an arbitrary random variable $G$ that is distributed according to $p_G$ over alphabet $\mc{G}$.  Formally, we will represent $G$ as a quantum state:
\begin{equation}
\label{Eq:RVstate}
\omega_G=\sum_{g\in\mc{G}}p(g)\op{g}{g},
\end{equation}
where the $\ket{g}$ are orthonormal vectors for a vector space of dimension $|\mc{G}|$.  When $\omega_G$ is subjected to classical operations (i)--(iv), it will transform as
\begin{align} \label{Eq:LOPCsteps}
\omega_G&\;\overset{(i)}{\longrightarrow}\;\omega_G\otimes \omega^{\text{local}}_S\qquad\text{for local ancillary variable $S$};\notag\\
\omega_G&\;\overset{(ii)}{\longrightarrow}\;\sum_{g\in\mc{G}}p(g)\op{g}{g}\otimes \op{g}{g}^{\text{local}};\notag\\
\omega_G&\;\overset{(iii)}{\longrightarrow}\;\sum_{g\in\mc{G}}p(g)\op{f(g)}{f(g)}^{\text{local}}\notag\\
\omega_G&\;\overset{(iv)}{\longrightarrow}\;\sum_{g\in\mc{G}}p(g)\op{g}{g}\otimes\op{f(g)}{f(g)}^{\text{global}}.
\end{align}
Here, ''local'' (resp. ``global'') refers to some system storing classical information which only one party (resp. all parties) can access.

\subsubsection{Quantum Operations.}

In the quantum LOPC setting, each party is allowed to perform the following operations:
\begin{enumerate}
\item[(i)]  Perform a local quantum instrument $(\mc{E}_m)_m$ \cite{Davies-1970a}, where each $\mc{E}_m$ is a completely positive (CP) map, and their sum $\sum_m \mc{E}_m$ is a trace-preserving map.  Quantum instruments represent the most general type of quantum measurement.  When performing the instrument on the state $\sigma$, the ``measurement'' outcome $m$ is obtained with probability $p(m)=tr[\mc{E}_m(\sigma)]$, and the post-measurement state given this outcome is $\sigma_m=\mc{E}_m(\sigma)/p(m)$.
\item[(ii)] Broadcast the result of any quantum measurement.
\end{enumerate}
A general \textbf{quantum LOPC protocol} $\mc{P}_q$ is described by a multi-level ``tree'' of local instruments in which the choice of instrument performed at each node of the tree depends on the particular history of measurement outcomes leading up to that node (see Ref. \cite{Kleinmann-2011a, Chitambar-2013c, Chitambar-2014b} for details).   Transformations (i)--(iii) in Eq. \eqref{Eq:LOPCsteps} fall within the framework of local quantum instruments since evaluating a function is a special type of quantum measurement in which the measurement outcome is the function's value.  Therefore, quantum LOPC generalizes the notion of classical LOPC.

With a slight abuse of notation, for a given classical/quantum protocol we use $\mc{P}_{c}/\mc{P}_q$ to denote both the particular protocol as well as the map associated with the protocol:
\[\mc{P}(\sigma^{ABE})=\sum_mp(m)\sigma^{ABE}_m\otimes \op{m}{m}^{\text{global}},\]
where $\sigma^{ABE}_m$ is the tripartite state generated when $m$ is the total broadcasted meassage.  The random variable associated with the global communication is \[\omega_M^{\text{global}}=\sum_mp(m)\op{m}{m}^{\text{global}}.\]

\subsection{A Classical versus quantum source of correlations}

Throughout this paper, we will assume that some basis for Alice, Bob, and Eve's system has chosen and is fixed.  Each of these is typically referred to as the \textit{computational basis} for the given system and is denoted by $\{\ket{x}^A\}_{x=1}^{d_A}$, $\{\ket{y}^B\}_{y=1}^{d_B}$ and $\{\ket{z}^E\}_{z=1}^{d_E}$ respectively.  Let $p_{XYZ}$ be an arbitrary three-way joint probability distribution for random variables $X$, $Y$, and $Z$ which takes on values $p(x,y,z)$.  We introduce the following physical instantiations of $p_{XYZ}$:
\begin{itemize}
\item A coherent embedding (or qqq embeding):
\begin{align}
\label{Eq:qqq}
\ket{\Psi_{qqq}}=&\sum_{x,y,z}e^{i\varphi_{xyz}}\sqrt{p(x,y,z)}\ket{xyz}^{ABE},\\
&\text{for any $\varphi_{xyz}\in[0,2\pi)$}.
\end{align}
\item A one-sided incoherent embedding (or cqq embedding):
\begin{align}
\rho_{cqq}&=\sum_{x}p(x)\op{x}{x}^A\otimes\op{\psi_{x}}{\psi_x}^{BE},\\ &\text{where $\ket{\psi_x}=\sum_{y,z}e^{i\varphi_{xyz}}\sqrt{p(y,z|x)}\ket{yz}$.}\notag
\end{align}
\item A two-sided incoherent embedding (or ccq embedding):
\begin{align}
\rho_{ccq}&=\sum_{x,y}p(x,y)\op{xy}{xy}^{AB}\otimes\op{\psi_{xy}}{\psi_{xy}}^{E},\\&\text{where $\ket{\psi_{xy}}=\sum_{z}e^{i\varphi_{xyz}}\sqrt{p(z|xy)}\ket{z}$.}\notag
\end{align}
\item An incoherent embedding (or ccc embedding):
\begin{equation}
\label{Eq:ccc}
\qquad\rho_{ccc}=\sum_{x,y,z}p(x,y,z)\op{xyz}{xyz}^{ABE}.
\end{equation}
\end{itemize}
Note that $\rho_{ccc}$ corresponds to the state $\omega_{XYZ}$ introduced in Eq. \eqref{Eq:RVstate}.  We can therefore think of $\rho_{ccc}$ as either a classical or quantum object, the difference being dictated by whether it is processed using either classical or quantum LOPC.

Regardless of how the phases $\phi_{xyz}$ are chosen, the various embeddings can be related through a series of local physical transformations:
\begin{align}
\label{Eq:dephase}
\op{\psi_{qqq}}{\psi_{qqq}}\overset{(1)}{\longrightarrow}\rho_{cqq}\overset{(2)}{\longrightarrow}\rho_{ccq}\overset{(3)}{\longrightarrow}\rho_{ccc},
\end{align}
where (1) is attained by Alice performing a dephasing channel
\[\sigma\to\sum_{x}\op{x}{x}\sigma\op{x}{x}\] and likewise for (2) and (3).  It is also worth pointing out that since we allow for phases $\phi_{xyz}$ in the coherent superposition of Eq. \eqref{Eq:qqq}, \textit{any tripartite pure state can be regarded as the qqq state of some distribution}.  This distribution $p_{XYZ}$ is obtained simply by all three parties dephasing in the computational basis as shown in Eq. \eqref{Eq:dephase}.  Our approach is thus more general to previous studies in coherent embeddings \cite{Christandl-2007a}, where it is typically assumed that $\phi_{xyz}=0$ for all $x$, $y$, and $z$.

\subsubsection{Secret Key Distillation.}

The scenario we consider is an identical, independent, and discrete (i.i.d.) source that is generating some particular embedding of $p_{XYZ}$ for Alice, Bob and Eve.  The goal of Alice and Bob is to distill secret key, which is shared randomness held independently of Eve's system.  We denote the state corresponding to $r:=\log s$ bits of perfectly shared randomness by
\begin{equation}
\Phi_{r}^{AB}=\frac{1}{s}\sum_{i=1}^s\op{ss}{ss}^{AB}.
\end{equation}
The notion of secret key rate is defined as follows.
\begin{definition}
For distribution $p_{XYZ}$, we say that $R$ is an \textbf{LOPC achievable key rate} if for every $\epsilon>0$, there exists a classical LOPC protocol $\mc{P}_c$ acting on $\sigma^{ABE}:=\rho_{ccc}^{\otimes n}$ (for $n$ sufficiently large) and generating messages $\omega_M^{\text{global}}$ such that 
\begin{equation}
\label{Eq:KeyrateDefn}
\left\|\mc{P}_c(\sigma^{AB})-\Phi_{\lfloor nR-\epsilon\rfloor}^{AB}\otimes \left(\omega_M^{\text{global}}\otimes \sigma^{E}\right)\right\|_1<\epsilon,
\end{equation}
where $\sigma^{AB}=tr_E\sigma^{ABE}$ and $\sigma^E=tr_{AB}\sigma^{ABE}$.  The supremum achievable key rate is denoted by $K_D(p_{XYZ})$.  We say that $R$ is a \textbf{ccc}, \textbf{ccq}, \textbf{cqq}, or \textbf{qqq LOPC achievable rate} if there exists a quantum LOPC protocol $\mc{P}_q$ to replace $\mc{P}_c$ in Eq. \eqref{Eq:KeyrateDefn}, and we further take $\sigma^{ABE}:=\rho_{ccc}^{\otimes n}$, $\sigma^{ABE}:=\rho_{ccq}^{\otimes n}$, $\sigma^{ABE}:=\rho_{cqq}^{\otimes n}$ or $\sigma^{ABE}:=\op{\psi_{qqq}}{\psi_{qqq}}^{\otimes n}$ respectively.  The supremum achievable key rates in these scenarios are denoted by $K_D(\rho_{ccc})$, $K_D(\rho_{ccq})$, $K_D(\rho_{cqq})$ and $K_D(\psi_{qqq})$ respectively.
\end{definition}


\subsection{Quantum Entanglement}

Quantum entanglement is a resource shared between two or more quantum systems that is distinct from secret key \cite{Horodecki-2009a}.  Starting from a tripartite pure state $\ket{\Psi_{qqq}}$, Alice and Bob share one entangled bit (ebit) of quantum information in the state $\ket{\Psi_{qqq}}$ if it has the form 
\begin{equation}
\label{Eq:Ebit}
\ket{\Psi_{qqq}}^{ABE}=\ket{\Phi_2}^{AB}\otimes \ket{\varphi}^E, 
\end{equation}
where $\ket{\Phi_2}^{AB}:=\sqrt{1/2}(\ket{00}+\ket{11})^{AB}$ is a so-called ebit and $\ket{\varphi}^E=\sum_ze^{i\varphi_z}\times\\\sqrt{p(z)}\ket{z}$ is any state held by Eve.  On the surface, the tripartite state $\ket{\Psi_{qqq}}^{ABE}=\ket{\Phi_2}^{AB}\otimes\ket{\varphi}^E$ looks very similar to the state $\rho^{ABE}_{ccq}=\Phi_2^{AB}\otimes \op{\varphi}{\varphi}^E$, which contains one bit distillable secret key and is obtained from $\ket{\Psi_{qqq}}^{ABE}$ through dephasing by Alice and Bob.  However there is a critical difference between the two states.  For $\Phi_2^{AB}\otimes \op{\varphi}{\varphi}^E$, it is entirely consistent that there should exist some third party Sapna ($S$) who holds as side information the value of Alice and Bob's bit in $\Phi_2^{AB}$.  In other words, we can envision a four-party state $\sigma^{ABES}=\Phi_3^{ABS}\otimes\op{\varphi}{\varphi}^E$ with $\Phi_3^{ABS}=\frac{1}{2}(\op{000}{000}+\op{111}{111})$.  And while $K_D(\sigma^{ABE})=1$, there is no secrecy with respect to Sapna: $K_D(\sigma^{ABS})=0$.  In contrast, Alice and Bob's entanglement in $\ket{\Psi_{ABE}}$ exists regardless of what side information is known.  That is, if $tr_{ES}(\sigma^{ABES})=\op{\Phi_2}{\Phi_2}^{AB}$ for \textit{any} state $\sigma^{ABES}$, then necessarily $\sigma^{ABES}$ has the product-state form $\sigma^{ABES}=\op{\Phi_2}{\Phi_2}^{AB}\otimes \sigma^{ES}$.  This means that if Alice and Bob should dephase when holding the state $\sigma^{ABES}$, they will generate key that is secret from not only Eve but also Sapna: $\op{\Phi_2}{\Phi_2}^{AB}\otimes \sigma^{ES}\to \Phi^{AB}_2\otimes \sigma^{ES}$.  Therefore, entanglement is a property of a bipartite $\rho^{AB}$ itself and, unlike secret key, one does not need to introduce any third party to speak of its entanglement.

Similar to the secret key rate $K_D$, one can also define for $\rho^{AB}$ the entanglement distillation rate $E_D$ \cite{Rains-1999a}.  This quantifies the asymptotic rate for which ebits can be obtained from $\rho^{AB}$ using local operations and classical communication (LOCC).  The operational class LOCC differs from quantum LOPC in that the former makes no explicit reference to a third party who records the ``public'' communication.  It is a fundamental and challenging problem in quantum information to compute $E_D(\rho^{AB})$ for a given quantum state.  Almost all meaningful measures of entanglement provide an upper bound for $E_D$ \cite{Horodecki-2009b}, and three such measures are the relative entropy of entanglement \cite{Vedral-1998a}, the squashed entanglement \cite{Christandl-2004a}, and the entanglement of formation \cite{Bennett-1996a}: 
\begin{itemize}
\item $E_r(\rho^{AB})$: the relative entropy of entanglement is  
\begin{equation}
E_r(\rho^{AB})=\min_{\sigma\in \mc{S}}S(\rho\||\sigma),
\end{equation}
where $\mc{S}$ is the set of separable density operators and $S(\rho||\sigma)=-tr[\rho\log\sigma]-S(\rho)$ is the relative entropy;
\item $E_{sq}(\rho^{AB})$: the squashed entanglement is 
\begin{equation}
E_{sq}(\rho^{AB})=\frac{1}{2}\inf_{\rho^{ABE}} I(A:B|E)_{\rho^{ABE}},
\end{equation}
where the infimum is taken over all extensions $\rho^{ABE}$ such that $\tr_E\rho^{ABE}=\rho^{AB}$, and $I(A:B|E)_{\rho^{ABE}}=S(AE)+S(BE)-S(ABE)-S(E)$ is the conditional quantum mutual information of the state $\rho^{ABE}$.  
\item $E_F(\rho^{AB})$: the entanglement of formation is 
\begin{equation}
\label{Eq:EOF}
E_F(\rho^{AB})=\min \sum_{i}p(i)S(\tr_A \varphi_i),
\end{equation}
with the minimization taken over all decompositions $\rho^{AB}=\sum_ip(i)\op{\varphi_i}{\varphi_i}$.
\end{itemize}
The particular significance of these entanglement measures is that they provide upper bounds not only for the distillable entanglement but also for distillable key.
\begin{theorem}
[\cite{Horodecki-2005d, Horodecki-2009b, Christandl-2006a}]\label{prop_HHC}
For an arbitrary tripartite state $\ket{\Psi_{qqq}}^{ABE}$ with \\$\rho^{AB}=tr_E\ket{\Psi_{qqq}}\bra{\Psi_{qqq}}^{ABE}$, the rates $K_D(\Psi_{qqq})$ and $E_D(\rho^{AB})$ are both upper bounded by the relative entropy of entanglement $E_r(\rho^{AB})$ as well as the squashed entanglement $E_{sq}(\rho^{AB})$.
\end{theorem} 
Unfortunately, each of the above entanglement measures involves a complicated minimization and in fact, their evaluation represents an NP-hard/NP-complete computational problem \cite{Huang-2014a}.  It is therefore not surprising that very few instances are known in which any of these measures can be explicitly computed.  In this paper, we introduce a new class of quantum states for which all these measures can be evaluated.  Our strategy will be based on the notion of \textit{reversible secrecy}, which we describe next.

\subsection{Reversible Entanglement and Secret Key}

Dual to the task of entanglement distillation is the task of entanglement formation, which describes building a given state $\rho^{AB}$ using LOCC and an initial supply of ebits.  The \textbf{entanglement cost} $E_C$ of a mixed state $\rho^{AB}$ is the asymptotic optimal rate of ebit consumption for Alice and Bob to generate faithful copies of $\rho^{AB}$ by LOCC \cite{Hayden-2001a}.  The entanglement cost is obviously lower bounded by the distillable entanglement, and compared to the above entanglement measures, the following hierarchy holds \cite{Horodecki-2009b}:
\begin{equation}
\label{Eq:EntCompare}
\begin{cases}E_D(\rho^{AB})\\K_D(\Psi_{qqq})\end{cases}\leq\begin{cases} E_{r}(\rho^{AB})\\E_{sq}(\rho^{AB})\end{cases}\leq E_C(\rho^{AB})\leq E_F(\rho^{AB}),
\end{equation}
where the first inequality references Theorem \ref{prop_HHC}.
A state $\rho^{AB}$ is said to possess \textbf{reversible entanglement} entanglement if $E_D(\rho^{AB})=E_C(\rho^{AB})$.  Operationally this means that the entanglement in $\rho^{AB}$ can be concentrated and diluted at equal rates.

Recently, the phenomenon of \textbf{reversible secrecy}, which is the classical analog to reversible entanglement, was studied in  \cite{Chitambar-2015a}.  Here, one first identifies the key cost $K_C$ of a distribution $p_{XYZ}$ as the amount of secret correlations needed for Alice and Bob to asymptotically prepare $p_{XYZ}$ using classical LOPC \cite{Renner-2003a}.  The distribution is said to possess reversible secrecy if $K_D(p_{XYZ})=K_C(p_{XYZ})$.  A key result that we prove in Theorem \ref{Thm:Main-Quant} is that when $p_{XYZ}$ is a reversible distribution and $\ket{\Psi_{qqq}}$ is a qqq embedding, then Eq. \eqref{Eq:EntCompare} can be further upper bounded as
\begin{equation}
\label{Eq:EntCompareReversible}
\begin{cases}E_D(\rho^{AB})\\K_D(\Psi_{qqq})\end{cases}\leq\begin{cases} E_{r}(\rho^{AB})\\E_{sq}(\rho^{AB})\end{cases}\leq E_C(\rho^{AB})\leq E_F(\rho^{AB})\leq K_D(p_{XYZ}).
\end{equation}
By identifying a class of distributions in Section \ref{sec_dist} for which $K_D(\Psi_{qqq})=K_D(p_{XYZ})$, this chain of inequalities becomes tight and we are thus able to compute the various entanglement measures of $\rho^{AB}$.

\subsection{Classes of Classical Distributions Related to Secrecy Reversibility} \label{sec_dist}

For a tripartite distribution $p_{XYZ}$, we can define \emph{maximal conditional common function} $J_{XY|Z}=\{J_{XY|Z=z}: p_Z(z)>0\}$, where $J_{XY|Z=z}$ is the common information of a bipartite distribution $P_{XYZ}(x,y|Z=z)$ \cite{Gacs-1973a}. For completeness, we provide a self-contained introduction of $J_{XY|Z}$ in Appendix~\ref{app_CI}. 
 \medskip
 
A distribution $p_{XYZ}$ is said to be (\cite{Chitambar-2014c, Chitambar-2015a}): 
\begin{itemize} 
\item \textbf{Block independent} (BI) if $I(X:Y|J_{XY|Z}Z)=0$.  

\item \textbf{Uniform block independent} (UBI) if both $I(X:Y|J_{XY|Z}Z)=0$ and
$H(J_{XY|Z}|X)=H(J_{XY|Z}|Y)=0.$

\item \textbf{Uniform block independent under public discussion} (UBI-PD) if it is BI and there is a public communication protocol generating messages $M$ such that $p_{(MX)(MY)(ZM)}$ is UBI and $I(M:J_{XY|Z}|Z)=0$.

\item \textbf{Uniform block independent under public discussion and eavesdropper's local processing} (UBI-PD$\downarrow$) if there exists a channel $\overline{Z}|Z$ such that $p_{XY|\overline{Z}}$ is UBI with the required public communication $M$ also satisfying $I(Z:J_{XY|\overline{Z}}|M\overline{Z})=0$.

\item \textbf{Semi-unambiguous} \cite{Christandl-2007a} if $H(Z|XY)=0$.

\item \textbf{Unambiguous} \cite{Ozols-2014a} if $H(Z|XY)=0$ and $H(XY|J_{XY|Z}Z)=0$.
\end{itemize}
The relations between these distributions are depicted in Figure~1.

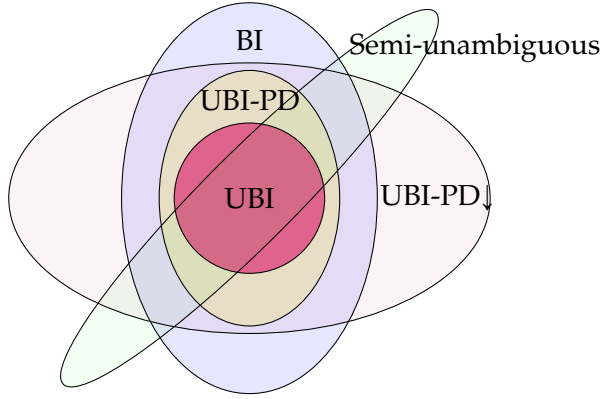
\begin{figure}\label{fig1}
\center
\begin{tikzpicture}
\def\homomorphism{(0,0) circle (1.0cm)}
\def\isomorphism{(0,0) ellipse (1.2cm and 1.7cm)}
\def\endomorphism{(0,0) ellipse (1.7cm and 2.6cm)}
\def\UBIPDDown{(0,0) ellipse (3.2cm and 1.8cm)}
\def\SUB{(-0,-0) ellipse (0.6cm and 3.5cm)}

      \begin{scope}[fill opacity=0.7]
        \fill[magenta] \homomorphism;
      \end{scope}

      \begin{scope}[fill opacity=0.1]
        \fill[blue] \endomorphism;
      \end{scope}

      \begin{scope}[fill opacity=0.2]
        \fill[yellow] \isomorphism;
      \end{scope}

      \begin{scope}[fill opacity=0.05]
        \fill[purple] \UBIPDDown;
      \end{scope}
      
      \begin{scope}[fill opacity=0.05]
        \fill[green][rotate=135] \SUB;
      \end{scope}

      \draw \homomorphism;
      \draw \isomorphism;
      \draw \endomorphism;
      \draw[rotate=135]  \SUB;
       \draw\UBIPDDown;
    {
        \scalefont{1.0}
        \node[text=black] at ( 0,0) {UBI};
      }

      {
        \scalefont{1.0}
        \node[text=black] at (0, 1.3) {UBI-PD};
      }

      {
        \scalefont{1.0}
        \node[text=black] at (  0, 2.1) {BI};
      }
      {
        \scalefont{1.0}
        \node[text=black] at (  2.5, 0) {UBI-PD$\downarrow$};
      }

      {
        \scalefont{1.0}
        \node[text=black] at (  3, 2) {Semi-unambiguous};
      }
    \end{tikzpicture}
    \caption{The relations between known classical distributions.}
\end{figure}

Being BI means that given $Z$, Alice and Bob share no more correlations besides their block number given by some maximal common function $J_{XY|Z}$.  For UBI distributions, the block number is independent of $Z$ and can therefore be computed locally by Alice and Bob.  Finally, for UBI-PD, the distribution becomes UBI once Alice shares with Bob her and Eve's common information, and Bob does likewise.  Of course, Eve learns nothing new about $XY$ through this public discussion. For semi-unambiguous distributions, the random variable $Z$ can be uniquely determined by random variables $X$ and $Y$; while for unambiguous distributions, each random variable can be uniquely determined by the other two random variables. 

The significance of these distributions to the problem of secrecy reversibility is given in the following theorem.
\begin{theorem}[\cite{Chitambar-2015a}]
\label{Thm:Reversible-Structure}
{\upshape(1)} If $K_C(p_{XYZ})=K_D(p_{XYZ})$ then there exists a channel for Eve $\overline{Z}|Z$ such that $p_{XY\overline{Z}}$ is BI.  {\upshape (2)} If $p_{XYZ}$ is UBI-PD$\downarrow$, then $K_C(p_{XYZ})=K_D(p_{XYZ})$.
\end{theorem}
In Ref. \cite{Chitambar-2015a}, it is shown that the necessary condition (1) and sufficient condition (2) are equivalent whenever either Alice or Bob holds a binary random variable.

\section{Main Results and Proofs}\label{sec_Main}

\subsection{Advantages in Quantum versus Classical Key Distillation}

\subsubsection{No advantages in incoherent embeddings.}

\begin{theorem}\label{thm:XYZ=ccc}
$K_D(p_{XYZ})=K_D(\rho_{ccc})$ for any distribution $p_{XYZ}$.
\end{theorem}
\begin{proof}
The inequality $K_D(\rho_{ccc})\geq K_D(p_{XYZ})$ is immediate from the fact every classical protocol $\mc{P}_c$ is a special type of quantum protocol $\mc{P}_q$.  Now we turn to the converse $K_D(\rho_{ccc})\geq K_D(p_{XYZ})$.  The idea will be to show that every quantum LOPC protocol $\mc{P}_q$ distilling secret key can be transformed into a classical protocol $\mc{P}_c$ that distills the same amount of key.  Suppose that $||\mc{P}_q(\rho_{ccc}^{\otimes n})-\Phi^{AB}_s\otimes(\omega_M^{\text{global}}\otimes \sigma^E)||<\epsilon$ with $\sigma^E=tr_{AB}(\rho_{ccc}^{\otimes n})$ and $\mc{P}_q$ some $r$-round quantum LOPC protocol.  To perform the following analysis let's fix some notation.  First, without loss of generality, let's assume that $r$ is even with Alice (resp. Bob) measuring in all the odd-numbered (resp. even-numbered) rounds.  We let $i_{\leq k}$ denote a particular sequence of the first $k$ rounds with $i_{<k}:=i_{\leq k-1}$.  If we wish to refer to a specific outcome in the $k^{th}$ round, we will denote this by $i_k$.  Hence $i_{\leq k}=(i_1,i_2,\cdots, i_k)$ for some particular sequence.  Finally, if, say, Alice is the measuring party in the $k^{th}$ round, we denote her local instrument conditioned on outcome $i_{<k}$ by $(\mc{A}^{(i<k)}_{i_k})_{i_k}$.  If we wish to speak of the full composition of Alice's CP maps corresponding to the outcome sequence $i_{\leq k}$, we will denote this simply by $\mc{A}^{(i_{\leq k})}$, with no subscript.  That is (for odd-numbered $k$) we have
\[\mc{A}^{(i_{\leq k})}=\mc{A}^{(i_{<k})}_{i_k}\circ\mc{A}^{(i_{\leq k-2})}=\mc{A}^{(i_{<k})}_{i_k}\circ \mc{A}^{(i_{<k-2})}_{i_{k-2}}\circ\cdots\circ\mc{A}_{i_1},\]
and similarly Bob's action is described by
\[\mc{B}^{(i_{\leq k-1})}=\mc{B}^{(i_{<k-1})}_{i_{k-1}}\circ\mc{B}^{(i_{\leq k-3})}=\mc{B}_{i_{k-1}}^{(i_{<k-1})}\circ\mc{B}_{i_{<k-3}}^{(i_{<k-3})}\circ\cdots\circ\mc{B}_{i_2}^{(i_1)}.\]

With the notation in hand, when performing protocol $\mc{P}_q$ on $\rho_{ccc}^{\otimes n}$, we can describe the overall state of generated across all outcome branches by
\begin{align}
&\mc{P}_q(\rho_{ccc}^{\otimes n})\notag\\
&=\sum_{\mbf{x},\mbf{y},\mbf{z}}\sum_{i_{\leq r}}p^n(\mbf{x},\mbf{y},\mbf{z})&\mc{A}^{(i_{\leq r-1})}\otimes\mc{B}^{(i_{\leq r})}(\op{\mbf{x}\mbf{y}}{\mbf{x}\mbf{y}})\otimes \op{\mbf{z}}{\mbf{z}}^E\otimes\op{i_{\leq r}}{i_{\leq r}}^{\text{global}},
\end{align}
where the first sum is over $\mc{X}^n\times\mc{Y}^n\times\mc{Z}^n$ with $p^n(\mbf{x},\mbf{y},\mbf{z})$ being the $n$-fold product distribution of $p_{XYZ}$, and the second sum is over all possible measurement sequences.  If Alice and dephase $\mc{P}_q(\rho_{ccc}^{\otimes n})$ in the computational basis, the resulting state will be at least $\epsilon$-close to $\Phi^{AB}_s\otimes(\omega_M^{\text{global}}\otimes \sigma^E)$ by monotonicity of the trace norm.  Hence it suffices to show that this dephased state $\Delta\left(\mc{P}_q(\rho_{ccc})\right)$ can be generated using classical LOPC.  To see that this is possible, we repeatedly use the fact that the messages are generated locally to form the expansion
\begin{align}
\Delta\left(\mc{P}_q(\rho_{ccc}^{\otimes n})\right)=\sum_{\mbf{x}',\mbf{y}'}\sum_{\mbf{x},\mbf{y},\mbf{z}}\sum_{i_{\leq r}}&Pr[\mbf{x}',\mbf{y}'|i_{\leq r},\mbf{x},\mbf{y},\mbf{z}]p^n(\mbf{x},\mbf{y},\mbf{z})\notag\\
&\cdot\op{\mbf{x}'\mbf{y}'}{\mbf{x}'\mbf{y}'}^{AB}\otimes\op{\mbf{z}}{\mbf{z}}^E\otimes\op{i_{\leq r}}{i_{\leq r}}^{\text{global}}
\end{align}
where
\begin{align}
Pr[\mbf{x}',\mbf{y}'|i_{\leq r},\mbf{x},\mbf{y},\mbf{z}]=\frac{\bra{\mbf{x}'}\mc{A}^{(i_{\leq r-1})}(\op{\mbf{x}}{\mbf{x}})\ket{\mbf{x}'}\cdot\bra{\mbf{y}'}\mc{B}^{(i_{\leq r})}(\op{\mbf{y}}{\mbf{y}})\ket{\mbf{y}'}}{Pr[i_{\leq r}|\mbf{x},\mbf{y},\mbf{z}]}
\end{align}
and
\begin{align}
Pr[i_{\leq r}|\mbf{x},\mbf{y},\mbf{z}]&=tr[\mc{A}^{(i_{\leq r-1})}\otimes\mc{B}^{(i_{\leq r})}(\op{\mbf{x}\mbf{y}}{\mbf{x}\mbf{y}})]\notag\\
&=\prod_{\text{even $k$}}^r\frac{tr[\mc{A}^{(i_{\leq k-1})}\otimes\mc{B}^{(i_{\leq k})}(\op{\mbf{x}\mbf{y}}{\mbf{x}\mbf{y}})]}{tr[\mc{A}^{(i_{\leq k-1})}\otimes\mc{B}^{(i_{\leq k-2})}(\op{\mbf{x}\mbf{y}}{\mbf{x}\mbf{y}}]}\notag\\
&\qquad\times\prod_{\text{odd $k$}}^r\frac{tr[\mc{A}^{(i_{\leq k})}\otimes\mc{B}^{(i_{\leq k-1})}(\op{\mbf{x}\mbf{y}}{\mbf{x}\mbf{y}})]}{tr[\mc{A}^{(i_{\leq k-2})}\otimes\mc{B}^{(i_{\leq k-1})}(\op{\mbf{x}\mbf{y}}{\mbf{x}\mbf{y}}]}\notag\\
&=\prod_{\text{even $k$}}^r\frac{tr[\mc{B}^{(i_{\leq k})}(\op{\mbf{y}}{\mbf{y}})]}{tr[\mc{B}^{(i_{\leq k-2})}(\op{\mbf{y}}{\mbf{y}}]}\times\prod_{\text{odd $k$}}^r\frac{tr[\mc{A}^{(i_{\leq k})}(\op{\mbf{x}}{\mbf{x}})]}{tr[\mc{A}^{(i_{\leq k-2})}(\op{\mbf{x}}{\mbf{x}}]}.
\end{align}
Thus a classical protocol $\mc{P}_c$ generating $\Delta(\mc{P}_q(\rho_{ccc}^{\otimes n}))$ is the following:
\begin{enumerate}
\item In the first round, Alice measures her variable $\mbf{x}$ and broadcasts message $i_1$ with probability $Pr[i_1|\mbf{x}]=tr[\mc{A}_{i_1}(\op{\mbf{x}}{\mbf{x}})]$.
\item In every subsequent even-numbered (resp. odd-numbered) round $k$, Bob (resp. Alice) consults the message history $i_{<k}$ and broadcasts $i_k$ with probability 
\begin{align}
Pr[i_k|i_{<k},\mbf{y}]&=\frac{tr[\mc{B}_{i_k}^{(i_{<k})}\circ\mc{B}^{(i_{\leq k-2})}(\op{\mbf{y}}{\mbf{y}})]}{tr[\mc{B}^{(i_{\leq k-2})}(\op{\mbf{y}}{\mbf{y}})]}\\
\bigg(\text{resp.}\quad Pr[i_k|i_{<k},\mbf{x}]&=\frac{tr[\mc{A}_{i_k}^{(i_{<k})}\circ\mc{A}^{(i_{\leq k-2})}(\op{\mbf{x}}{\mbf{x}})]}{tr[\mc{A}^{(i_{\leq k-2})}(\op{\mbf{x}}{\mbf{x}})]}\bigg).
\end{align}
\item At the end of $r$ rounds with the total message $i_{\leq r}$ having been generated, Alice and Bob process their variables using local channels $\mbf{x}\to\mbf{x}'$ and $\mbf{y}\to\mbf{y}'$ with transition probabilities given by
\begin{align}
Pr[\mbf{x}'|i_{\leq r},\mbf{x}]&=\frac{\bra{\mbf{x}'}\mc{A}^{(i_{\leq r-1})}(\op{\mbf{x}}{\mbf{x}})\ket{\mbf{x}'}}{tr[\mc{A}^{(i_{\leq r-1})}(\op{\mbf{x}}{\mbf{x}})]},\\
Pr[\mbf{y}'|i_{\leq r},\mbf{y}]&=\frac{\bra{\mbf{y}'}\mc{B}^{(i_{\leq r})}(\op{\mbf{y}}{\mbf{y}})\ket{\mbf{y}'}}{tr[\mc{B}^{(i_{\leq r})}(\op{\mbf{y}}{\mbf{y}})]}.
\end{align}
\item It can be seen that the state generated through this process is precisely $\Delta(\mc{P}_q(\rho_{ccc}^{\otimes n}))$.
\end{enumerate}
\end{proof}

\subsubsection{Arbitrarily large advantages in coherent embeddings.}

\begin{theorem}
\label{Prop:KDEFgap1}
For any $N$, a distribution $p_{XYZ}$ exists that such that when embedding $p_{XYZ}$ into a coherent quantum source, one of the following relationships holds:
\begin{itemize}
\item[(a)] Eve gains an arbitrarily large advantage: $K_D(p_{XYZ})-K_D(\Psi_{qqq}^{ABE})>N$, or
\item[(b)] Alice and Bob gain an arbitrarily large advantage: $K_D(\Psi_{qqq}^{ABE})-K_D(p_{XYZ})>N$.
\end{itemize}
\end{theorem}
\begin{remark}
In the proof of (a) we actually demonstrate a much stronger result that $K_D(p_{XYZ})-E_F(\rho^{AB})>N$.  This means that we can distill key from $p_{XYZ}$ from a considerably higher rate than the rate of entanglement needed to generate the corresponding quantum state $\rho^{AB}$.  To our knowledge, this is the first known result of its kind.
\end{remark}
\begin{proof}
(a) We consider a very simple binary distribution with $p(0,0,0)=p(1,1,0)=1/4$, $p(0,0,1)=\lambda/2$, and $p(1,1,1)=(1-\lambda)/2$.  This is a UBI distribution and by Theorem \ref{Thm:Reversible-Structure} below, its key rate precisely $K_D(p_{xyz})=[1+h(\lambda)]/2$, where $h(x)=-x\log x-(1-x)\log(1-x)$.  The corresponding qqq embedding has the form
\[\ket{\Psi_{qqq}^{ABE}}=\sqrt{1/2}[\ket{\Phi}^{AB}\ket{0}^E+(\sqrt{\lambda}\ket{00}+\sqrt{1-\lambda}\ket{11})^{AB}\ket{1}^E].\]
Since $\rho^{AB}$ is a two-qubit state, its entanglement of formation can be calculated using the celebrated concurrence formula \cite{Wootters-1998a} (see also \cite{Chitambar-2015a}), and it is found to be 
\[E_F(\rho^{AB})=h\left([1+\sqrt{1-\left(1+\sqrt{\lambda(1-\lambda)}\right)^2)}]/2\right).\]
A simple convexity argument shows that $K_D(p_{XYZ})>E_F(\rho^{AB})$ whenever $0<\lambda<1/2$.  We now consider $n$ copies of $p_{XYZ}$.  Inspection reveals that $p_{XYZ}^{\otimes n}$ is also UBI for any $n$.  Thus, $K_D(p_{XYZ})=n[1+h(\lambda)]/2$.  On the other hand, the entanglement of formation is a sub-additive quantity such that $E_F\left((\rho^{AB})^{\otimes n}\right)\leq n E_F(\rho^{AB})$.  Consequently, for any $0<\lambda<1/2$ we attain an arbitrarily large gap between $K_D(p_{XYZ}^n)$ and $E_F\left((\rho^{AB})^{\otimes n}\right)$.  By Theorem \ref{prop_HHC} and the fact that $E_F(\rho^{AB}\geq E_{sq}(\rho^{AB})$, this gap will be at least as large as the gap between $K_D(p_{XYZ}^n)$ and $K_D\left((\Psi_{qqq}^{ABE})^{\otimes n}\right)$.

(b)  Consider the state $\ket{\Psi}^{ABE}=\sqrt{1/2}(\ket{00}+\ket{1+})^{AB}\ket{0}^E$ where Eve is initially uncorrelated.  This is a qqq encoding of a distribution $p_{XYZ}=p_{XY}\cdot q_Z$ whose mutual information is $I(X:Y)=1-h(1/3)\approx .311$.   Since Eve has no side information, the classical secret key rate $K_D(p_{XYZ})$ is equal to the mutual information, a well-known result in secret key agreement \cite{Maurer-1993a, Ahlswede-1993a}.  On the entanglement side, the reduced-state entropy characterizes the entanglement distillation rate for pure states \cite{Bennett-1996b}; hence $S(\rho^B)=E_D(\Psi^{ABE})$.  One bit of entanglement can be converted into one bit of secret key, and thus $K_D(\Psi^{ABE})\geq S(\rho_B)$.  In fact, this inequality is tight since $S(\rho^B)=E_{sq}(\Psi^{AB})\geq K_D(\Psi^{ABE})$.   Because both the mutual information and von Neumann entropy are additive, a similar argument to part (a) shows that the gap between $S(\left(\rho^{B}\right)^{\otimes n})-K_D(p_{XYZ})$ can be made arbitrarily large. 
\end{proof} 


\subsubsection{Advantages in one-sided versus two-sided coherent embeddings.}

\begin{center}\vspace{-.7cm}\end{center}

From Theorem \ref{thm:XYZ=ccc}, we see that the presence of quantum coherence is necessary for there to be a difference in quantum and classical key distillation.  The next example extends this result to states in which only one of the subsystems has quantum coherence.
\begin{lemma}
There exists states for which
\[K_D(\Psi_{qqq})>K_D(\rho_{cqq})>K_D(\rho_{ccq}).\]
In fact, the gaps between these rates can be arbitrarily large.
\end{lemma}
\begin{proof}
Consider the qqq embedded state
\begin{align}
\ket{\Psi_{qqq}}=\frac{1}{\sqrt{6}}\bigg([\ket{00}+\ket{11}]\otimes\ket{0}+[\ket{+2}+\ket{-3}]\otimes\ket{1}+[\ket{2+}+\ket{3-}]\otimes\ket{2}\bigg)
\end{align}
where $\ket{\pm}=\frac{1}{\sqrt{2}}(\ket{0}\pm\ket{1})$.  The two dephased states of interest are given by
\begin{align}
\rho_{cqq}&=\frac{1}{3}\Phi_2^{AB}\otimes\op{0}{0}^E+\frac{1}{6}(\op{0}{0}+\op{1}{1})^A\otimes(\op{2}{2}+\op{3}{3})^B\otimes\op{1}{1}^E\notag\\
&\qquad\qquad\qquad\qquad+\frac{1}{6}(\op{2+}{2+}+\op{3-}{3-})^{AB}\otimes\op{2}{2}^E,\\
\rho_{ccq}&=\frac{1}{3}\Phi_2^{AB}\otimes\op{0}{0}^E+\frac{1}{6}(\op{0}{0}+\op{1}{1})^A\otimes(\op{2}{2}+\op{3}{3})^B\otimes\op{1}{1}^,\notag\\
&\qquad\qquad\qquad\qquad+\frac{1}{6}(\op{2}{2}+\op{3}{3})^{A}\otimes(\op{0}{0}+\op{1}{1})^{B}\otimes\op{2}{2}^E.
\end{align}
From the squashed entanglement upper bound on $K_D$, we find that $K_D(\Psi_{qqq})\leq 1$, $K_D(\rho_{cqq})\leq 2/3$, and $K_D(\rho_{ccq})\leq 1/3$.  In fact, each of these bounds are tight using the following protocol.  For instance, with the state $\ket{\Psi_{qqq}}$, Alice and Bob both perform a projection into the $\{\ket{0},\ket{1}\}$ or $\{\ket{2},\ket{3}\}$ subspace.  If one of them indeed projects into $\{\ket{2},\ket{3}\}$, then the other party dephases in the $\{\ket{+},\ket{-}\}$ basis.  Finally, a local unitary rotation is made back to the $\{\ket{0},\ket{1}\}$ basis for both parties.  Doing so generates the state $\Phi_2^{AB}\otimes(\op{0}{0}+\op{1}{1}+\op{2}{2})^E/3$.  For $\rho_{cqq}$ a similar protocol is performed except that no key is obtained when Bob projects into the $\{\ket{2},\ket{3}\}$ subspace (which occurs with probability $1/3$).  And for $\rho_{ccq}$, secret key is obtained with probability $1/3$ when both parties projection into the $\{\ket{0},\ket{1}\}$ subspace.  Note that for $n$ copies, the upper bounds become $K_D(\Psi^{\otimes n}_{qqq})\geq n$, $K_D(\rho_{cqq}^{\otimes n})\geq 2/3 n$, and $K_D(\rho_{ccq}^{\otimes n})\geq 1/3 n$, which again are all tight.  So by considering many copies, we obtain examples where the gap between the different key rates is arbitrarily large.

\end{proof}

\subsection{Embedding Distributions with Reversible Secrecy}
\label{sec_ER}

We now consider qqq embeddings of distributions with reversible secrecy, for which it turns out that the quantum embedding favours Eve over Alice and Bob.  When adding UBI-PD and/or semi-unambiguous structure, relationships between key and quantum entanglement can be drawn.  
\begin{theorem}
\label{Thm:Main-Quant}
\begin{itemize}
\item[(a)]  If $p_{XYZ}$ has reversible secrecy (i.e. $K_C(p_{XYZ})=K_D(p_{XYZ})$), then 
\begin{equation}
K_D(p_{XYZ})\geq E_{sq}(\rho^{AB}).
\end{equation}
\item[(b)] If $p_{XYZ}$ is UBI-PD (and hence reversible), then 
\begin{equation}
K_D(p_{XYZ})\geq E_F(\rho^{AB}).
\end{equation}
\item[(c)] If $p_{XYZ}$ has reversible secrecy and is semi-unambiguous, then
\begin{equation}
K_C(p_{XYZ})=E_{sq}(\rho^{AB})=K_D(\Psi^{ABE})=K_D(p_{XYZ}).
\end{equation}
\item[(d)] If $p_{XYZ}$ has reversible secrecy, is semi-unambiguous and UBI-PD, then
\begin{multline*}
K_C(p_{XYZ})=K_D(p_{XYZ})=K_D(\Psi^{ABE})=E_F(\rho^{AB})\\=E_C(\rho^{AB})=E_D(\rho^{AB})=E_r(\rho^{AB})=E_{sq}(\rho^{AB})= H(J_{XY|Z}|Z).
\end{multline*}
\end{itemize}
\end{theorem}
\begin{proof} 

(a) By Theorem \ref{Thm:Reversible-Structure}, if $p_{XYZ}$ is reversible then there must be a channel $\overline{Z}|Z$ such that $p_{XY\overline{Z}}$ is block independent.  In other words, there exists a decomposition 
\begin{equation}
p_{XY|\overline{Z}}(x,y|\overline{z})=\sum_{j_{\overline{z}}}p(x|j_{\overline{z}},\overline{z})p(y|j_{\overline{z}},\overline{z})p(j_{\overline{z}}|\overline{z})
\end{equation}
where $p_X(\cdot|j_{\overline{z}},\overline{z})$ and $p_X(\cdot|j_{\overline{z}}',\overline{z})$ are disjoint distributions for $j_{\overline{z}}\not=j_{\overline{z}}'$, and likewise for Bob's conditional distributions.  For each $\overline{z}$, define the local measurement channel acting on Alice's system
\begin{align}
\omega\mapsto\Omega^{(\overline{z})}_A(\omega)&:= \sum_{j_{\overline{z}}}\sum_{x\;\text{such that} \atop{\;p(x|j_{\overline{z}}},\overline{z})>0}\bra{x}\omega\ket{x}\op{j_{\overline{z}}}{j_{\overline{z}}}.
\end{align}
Let $\Omega^{(\overline{z})}_B$ be defined similarly for Bob's system.

 We next consider the decomposition $\ket{\Psi_{ABE}}=\sum_{z}\sqrt{p(z)}\ket{\varphi_{z}}\ket{z}$, in which $\ip{xy}{\varphi_z}=\sqrt{p(x,y|z)}e^{i\varphi_{xyz}}$.  Note that $\rho^{AB}=tr_{\overline{Z}}\sigma^{AB\overline{Z}}$ where
\begin{equation}
\label{Eq:AB-extension}
\sigma_{AB\overline{Z}}:=\sum_{\overline{z}}\sum_{z}p(z|\overline{z})p(\overline{z})\op{\varphi_z}{\varphi_z}\otimes\op{\overline{z}}{\overline{z}}=\sum_{\overline{z}}p(\overline{z})\sigma^{AB}_{(\overline{z})}\otimes\op{\overline{z}}{\overline{z}}.
\end{equation}
On the state $\sigma^{AB\overline{Z}}$, we first dephase in the computational basis, and then apply $\Omega^{(\overline{z})}_A\otimes\Omega^{(\overline{z})}_B$ conditioned on $\overline{z}$.  Doing so generates the state
\begin{align}
\widehat{\sigma}^{AB\overline{Z}}&:=\sum_{\overline{z}}p(\overline{z})\sum_{x,y}p(x,y|\overline{z})\Omega^{(\overline{z})}_A(\op{x}{x})\otimes\Omega^{(\overline{z})}_B(\op{y}{y})\otimes\op{\overline{z}}{\overline{z}}\notag\\
&=\sum_{\overline{z}}p(\overline{z})\sum_{j_{\overline{z}}}p(j_{\overline{z}}|\overline{z})\op{j_{\overline{z}}j_{\overline{z}}\overline{z}}{j_{\overline{z}}j_{\overline{z}}\overline{z}}.
\end{align}
Hence,
\begin{align}
E_{sq}(\rho_{AB})&\leq \tfrac{1}{2}\sum_{\overline{z}}p(\overline{z})I(A:B)_{\sigma^{AB}_{(\overline{z})}}\notag\\
&\leq \tfrac{1}{2}\sum_{\overline{z}}p(\overline{z})[S(\sigma^A_{(\overline{z})})+S(\sigma^B_{(\overline{z})})]\notag\\
&\leq \tfrac{1}{2}\sum_{\overline{z}}p(\overline{z})[S(\widehat{\sigma}^A_{(\overline{z})})+S(\widehat{\sigma}^B_{(\overline{z})})]\notag\\
&=\sum_{\overline{z}}p(\overline{z})H(J_{\overline{Z}}|\overline{Z}=\overline{z})\notag\\
&= K_D(p_{XYZ}).
\label{Eq:squash1}
\end{align}

(b) If $p_{XYZ}$ is UBI-PD then again by Theorem \ref{Thm:Reversible-Structure}, we have $K_C(p_{XYZ})=K_D(p_{XYZ})=H(J_{XY|Z}|Z)$.  
Eq.~(\ref{Eq:AB-extension}) still holds in this case with $\overline{Z}=Z$, and $\sigma_{ABZ}$ is obtained from $\ket{\Psi_{ABE}}$ by Eve dephasing in the computational basis.  Block-independence of $p_{XYZ}$ means that $S(tr_A\op{\varphi_z}{\varphi_z})=H(J_{XY|Z}|Z=z)$.  Since $\{p(z),\ket{\varphi_z}\}$ provides a pure-state ensemble realizing $\rho^{AB}$, we see that $H(J_{XY|Z}|Z)\geq E_F(\rho^{AB})$.

(c)  From Theorem \ref{prop_HHC} and Ref. \cite{Christandl-2007a}, semi-unambiguous distributions are shown to satisfy the inequality $E_{sq}(\rho^{AB})\geq K_D(\Psi_{qqq}^{ABE})\geq K_D(p_{XYZ})$.  Combining with part (a) gives the desired result.

(d)  This follows from combining (c), (d) and Theorem \ref{prop_HHC} with the fact that $E_F(\rho^{AB})\geq E_C(\rho^{AB})\geq \max\{E_r(\rho^{AB}),E_{sq}(\rho^{AB})\}$.

\end{proof}

\begin{remark}
While Theorem \ref{Thm:Main-Quant} (a) implies that Alice and Bob never gain an advantage over Eve when embedding their correlations into a quantum source, it is not difficult to construct distributions in which Eve gains a non-zero advantage in the quantum setting.  This can be seen by the chain of inequalities in Eq. \eqref{Eq:squash1}.  In particular, whenever $\sigma_{\overline{z}}^{AB}$ is not pure the inequality will be strict.  This will hold, for instance, for any distribution $p_{XYZ}$ with a non-trivial channel $\overline{Z}|Z$ such that $p_{XY\overline{Z}}$ is UBI-PD.
\end{remark}

\section{Conclusion}

In this paper we have considered the task of resource distillation in quantum and classical sources.  Since secret key can be distilled from both quantum and classical states, a direct comparison can be made between the two scenarios.   Quantum states that are diagonal in some fixed basis - such as $\rho_{ccc}$ of Eq. \eqref{Eq:ccc} - lack coherence and are typically referred to as ``classical'' states since they possess the same entropic properties as classical probability distributions.  However, as quantum objects, these states can undergo certain physical transformations that are impossible for classical states.  We have shown that despite this enhanced dynamical ability, secret key distillation is equivalent for a classical distribution and its incoherent quantum embedding.  The situation is much different when the embedding takes the form of a coherent superposition.  We have presented examples when quantum and classical key rates can be vastly different; sometimes it benefits Alice and Bob to have their correlations embedded in a quantum state, sometimes it benefits Eve.  We have linked this investigation of quantum advantages to the problem of LOPC secrecy reversibility.  By introducing different families of distributions that demonstrate secrecy reversibility, we are able to compute the entanglement and distillable key of the embedded quantum states.  It is quite beautiful that notoriously difficult entanglement measures can be computed using exclusively a classical analysis of the states's underlying probability distribution.  We hope this paper helps advance our understanding of the relationship between classical secrecy and quantum entanglement.

\bibliographystyle{alphaurl}
\bibliography{QuantumBib}

\section{Appendix}

\label{Appendix}

\subsection{Common Information}
\label{app_CI}
Let $X$ and $Y$ be random variables over finite sets $\mc{X}$ and $\mc{Y}$ respectively, with joint distribution $p_{XY}$. A \textbf{common partitioning of length $t$} for $X$ and $Y$ are pairs of subsets $(\mc{X}_i,\mc{Y}_i)_{i=1}^t$ such that 
\begin{itemize}
\item[(i)] $\mc{X}_i\cap\mc{X}_j=\mc{Y}_i\cap \mc{Y}_j=\emptyset$ for $i\not=j$, 
\item[(ii)] $p(\mc{X}_i|\mc{Y}_j)=p(\mc{Y}_i|\mc{X}_j)=\delta_{ij}$, and  
\item[(iii)] if $(x,y)\in\mc{X}_i\times \mc{Y}_i$ for some $i$, then $p_X(x)p_Y(y)>0$. 
\end{itemize}
For every common partitioning, we associate a random variable $J(X,Y)$ such that $J(x,y)=i$ if $(x,y)\in\mc{X}_i\times\mc{Y}_i$.  Note that each party can determine the value of $J$ from their local information, and it is therefore called a \textbf{common function} of $X$ and $Y$ \cite{Gacs-1973a}.  We will refer to the sets $\mc{X}_i\times\mc{Y}_i$ as the \textbf{blocks} of the given partitioning, and hence $J$ is the ``block number'' random variable for the particular common partitioning.  

A \textbf{maximal common partitioning} is a common partitioning of greatest length.  
\begin{proposition}
\label{Prop:Partition-Unique}
Every pair of finite random variables $X$ and $Y$ has a unique maximal common partitioning.
\end{proposition}

With Proposition \ref{Prop:Partition-Unique}, we can speak unambiguously of \textit{the} maximal common partitioning of two random variables $X$ and $Y$.  We let $J_{XY}$ denote a \textbf{maximal common function} associated with the maximal common partitioning of $XY$; i.e. $J_{XY}(x,y)$ is the number $i$ for which $(x,y)\in\mc{X}_i\times\mc{Y}_i$ in the partitioning.  Since $J_{XY}$ is determined entirely from either $X$ or $Y$, with a slight abuse of notation we will sometimes write $J_{XY}(x)$ (resp. $J_{XY}(y)$) if we wish $J_{XY}$ to denote a function of $x$ (resp. $y$).  Note that $J_{XY}$ is unique up to a relabeling of the blocks $\mc{X}_i\times\mc{Y}_i$.  The entropic quantity $H(J_{XY})$ was identified by G\'{a}cs and K\"{o}rner as the \textbf{common information} of $X$ and $Y$ \cite{Gacs-1973a}.

We now turn to the conditional form of common functions.  For distribution $p_{XYZ}$, the variable $J_{XY|Z}$ is called a \textbf{maximal conditional common function} if $J_{XY|Z=z}$ is a maximal common function of the conditional distribution $p_{XY|Z=z}$.  Explicitly, we have $J_{XY|Z}(x,y,z)=j$ if $(x,y)\in\mc{X}_j\times\mc{Y}_j$ for some fixed block labeling of the maximal common   partitioning of $p_{XY|Z=z}$.  Again, note that $J_{XY|Z}$ is unique for the distribution $p_{XYZ}$ up to the choice of block labeling for each conditional distribution $p_{XY|Z=z}$.  For the case when Alice and Bob are perfectly correlated (or when there is just one party), we also define $J_{X|Z}:=J_{XX|Z}$.

\end{document}